\documentclass[11pt,draftcls,onecolumn,twoside]{IEEEtran} 
\usepackage{epsfig}
\usepackage{amsmath}
\usepackage{amstext}
\usepackage{amsfonts}
\usepackage{amssymb}
\usepackage{eucal}
\usepackage{graphicx}
\usepackage{verbatim}
\usepackage{stmaryrd}
\usepackage{bm}
\usepackage{etex}
\usepackage{easybmat}
\usepackage{pstricks}
\usepackage{pst-node}
\usepackage{psfrag}
\usepackage{tabularx}
\usepackage{float}

\newenvironment{proofsketch}{\trivlist\item[]\emph{ Sketch of proof}:}%
{\unskip\nobreak\hskip 1em plus 1fil\nobreak$\Box$
\parfillskip=0pt%
\endtrivlist}

\usepackage{tikz}
\usepackage{pgfplots}
\usepackage{algorithmic}

\newcommand{\tr}{\mathrm{Tr}\,}

\newcommand{\mrnhat}{m_{\underline{\hat{\bf R}}_N}}
\newcommand{\mn}{\underline{m}_{N}}
\newcommand{\rnhat}{\hat{\underline{\bf R}}_N }
\newtheorem{theorem}{Theorem}
\newtheorem{lemma}{Lemma}
\newtheorem{proposition}{Proposition}

\newtheorem{remark}{Remark}
\newtheorem{assumption}{\bf{Assumption}}

\begin{document}
\bibliographystyle{IEEEtran}
\title{Estimation of the Covariance Matrix of Large Dimensional Data} 
\author{Jianfeng Yao, Abla Kammoun, Jamal Najim\\
\today
\thanks{Yao, Kammoun, Najim are with T\'el\'ecom Paristech, France.}

\thanks{{\tt \{yao,kammoun,najim\}@telecom-paristech.fr}\ .}
\thanks{Najim is also in Centre Nationale de la Recherche Scientifique (CNRS). }
}

\maketitle

\begin{abstract}

  This paper deals with the problem of estimating the covariance
  matrix of a series of independent multivariate observations, in the
  case where the dimension of each observation is of the same order as
  the number of observations. Although such a regime is of interest for
  many current statistical signal processing and wireless
  communication issues, traditional methods fail to produce consistent
  estimators and only recently results relying on large random matrix
  theory have been unveiled.

  In this paper, we develop the parametric framework proposed by
  Mestre, and consider a model where the covariance matrix to be estimated has a (known)
  finite number of eigenvalues, each of it with an unknown
  multiplicity. The main contributions of this work are essentially
  threefold with respect to existing results,
  and in particular to Mestre's work: To relax the (restrictive)
  separability assumption, to provide joint consistent estimates for
  the eigenvalues and their multiplicities, and to study the variance
  error by means of a Central Limit theorem.


\end{abstract}

\section{Introduction}

Estimating the covariance matrix of a series of independent
multivariate observations is a crucial issue in many signal processing
applications. A reliable estimate of the covariance matrix is for
instance needed in principal component analysis \cite{Jollife86},
direction of arrival estimation for antenna arrays \cite{SCH86}, blind
subspace estimation \cite{moulines95}, capacity estimation
\cite{FOS98}, estimation/detection procedures \cite{SCH86,WAX85}, etc.

In the case where the dimension $N$ of the observations is small
compared to the number $M$ of observations, then the empirical
covariance matrix based on the observations often provides a good
estimate for the unknown covariance matrix. This estimate becomes
however much less accurate, and even not consistent with the dimension
$N$ getting higher (see for instance \cite[Theorem 2]{MES08}).

An interesting theoretical framework for modern estimation of
multi-dimensional variables occurs whenever the number of available
samples $M$ grows at the same pace as the dimension $N$ of the
considered variables. Shifting to this new assumption induces
fundamental differences in the behavior of the empirical covariance
matrix as analyzed in Mestre's work \cite{MES08,MES08b}.  Recently,
several attempts have been done to address this problem
(cf. \cite{MES08,MES08b,KAR08,CHE10}) using large random matrix theory which
proposed powerful tools, mainly spurred by the G-estimators of Girko \cite{GIR00},
to cope with this new context. This was for instance the
main ingredient used in \cite{KAR08} and \cite{ledoit}, where
grid-based techniques for inverting the Mar\v{c}enko-Pastur equation were
proposed.

In this article, we shall consider the case where the dimension of each observation $N$ together with the sample dimension $M$ go to infinity at the same pace, i.e. that their ratio converges to some nonnegative constant $c>0$. In order to present the contribution provided in this paper, let us describe the model under study.

Consider a $N\times M$ matrix ${\bf X}_N=(X_{ij})$ whose entries are
independent and identically distributed (i.i.d.) random variables. Let ${\bf R}_N$ be a $N\times N$ Hermitian
matrix with $L$ ($L$ being fixed and known) distinct eigenvalues $0<\rho_1<\cdots<
\rho_L$ with respective multiplicities $N_1,\cdots, N_L$ (notice that
$\sum_{i=1}^L N_i = N$). Consider now
$$
{\bf Y}_N = {\bf R}_N^{1/2} {\bf X}_N\ .
$$
The matrix ${\bf Y}_N=[{\bf y}_1,\cdots,{\bf y}_M]$ is the
concatenation of $M$ independent observations, where each observation
writes ${\bf y}_i = {\bf R}_N^{1/2}{\bf x}_i$ with ${\bf X}_N= [{\bf
  x}_1,\cdots, {\bf x}_M]$.  In particular, the covariance matrix of
each observation ${\bf y}_i$ is ${\bf R}_N =\mathbb{E} {\bf y}_i {\bf
  y}_i^H$ (matrix ${\bf R}_N$ is sometimes called the population
covariance matrix).


We consider the problem of estimating individually the eigenvalues $\rho_i$ as well as their multiplicities $N_i$.  Among the proposed parametric techniques, we cite the one developed by Mestre \cite{MES08b} and taken up by Vallet {\it et al} \cite{LOU10} and Couillet {\it et al} \cite{COU10b} for more elaborated models. Although being computationally efficient, this technique requires a {\em separability condition}, namely the assumption that the number of samples is large compared to the dimension of each sample (small limiting ratio $c=\lim \frac NM >0$). In such a case, the limiting spectrum of the empirical covariance matrix possesses as many clusters\footnote{By cluster, we mean a connex component of the support of the limiting probability distribution of the spectrum.}as there are eigenvalues to be estimated, and each eigenvalue can be estimated by a contour integral surrounding the related cluster. Mestre's technique cannot be applied anymore in the case where $c$ is larger (which reflects a higher dimension of the observation dimension with respect to the sample dimension). In fact, the dimension of the clusters may grow and neighbouring clusters may merge, violating the one-to-one correspondence between clusters and eigenvalues to be estimated (see for instance Fig. \ref{fig:clusters1} and \ref{fig:clusters2}). 

\begin{minipage}[c]{0.43\textwidth}
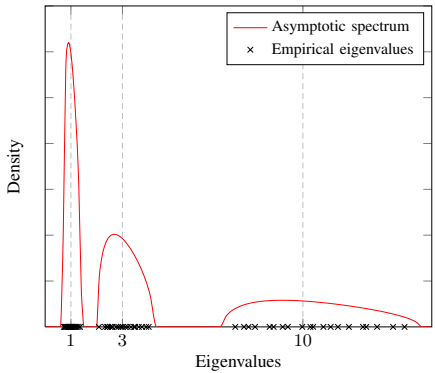
\begin{figure}[H]
  \begin{center}
  \begin{tikzpicture}[font=\footnotesize,scale=0.75]
    \renewcommand{\axisdefaulttryminticks}{8}
    \tikzstyle{every major grid}+=[style=densely dashed]
\pgfplotsset{every axis y label/.append style={yshift=-20pt}}
\pgfplotsset{every axis x label/.append style={yshift=5pt}}
    \tikzstyle{every axis legend}+=[cells={anchor=west},fill=white,
        at={(0.98,0.98)}, anchor=north east, font=\scriptsize ]

    \begin{axis}[
      xmajorgrids=true,
      xlabel={Eigenvalues},
      ylabel={Density},
      xtick = {1,3,10},
      ytick = {},
      yticklabels = {},
      xmin=0,
      xmax=15,
      ymin=0,
      ymax=0.07,
      ]

      \addplot[smooth,red,line width=0.5pt] plot coordinates{
(0,0) (0.301000,-0.000000) (0.401000,-0.000000) (0.501000,-0.000000) (0.601000,-0.000000) (0.701000,0.019786) (0.801000,0.056237) (0.901000,0.061934) (1.001000,0.059629) (1.101000,0.053338) (1.201000,0.044219) (1.301000,0.031803) (1.401000,0.008620) (1.501000,-0.000000) (1.601000,-0.000000) (1.701000,-0.000000) (1.801000,0.000000) (1.901000,0.000000) (2.001000,0.000000) (2.101000,0.010929) (2.201000,0.015193) (2.301000,0.017541) (2.401000,0.018929) (2.501000,0.019721) (2.601000,0.020103) (2.701000,0.020184) (2.801000,0.020037) (2.901000,0.019709) (3.001000,0.019234) (3.101000,0.018636) (3.201000,0.017929) (3.301000,0.017125) (3.401000,0.016229) (3.501000,0.015242) (3.601000,0.014157) (3.701000,0.012965) (3.801000,0.011643) (3.901000,0.010150) (4.001000,0.008405) (4.101000,0.006210) (4.201000,0.002623) (4.301000,-0.000000) (4.401000,-0.000000) (4.501000,-0.000000) (4.601000,-0.000000) (4.701000,-0.000000) (4.801000,-0.000000) (4.901000,-0.000000) (5.001000,-0.000000) (5.101000,-0.000000) (5.201000,-0.000000) (5.301000,0.000000) (5.401000,0.000000) (5.501000,0.000000) (5.601000,0.000000) (5.701000,0.000000) (5.801000,0.000000) (5.901000,0.000000) (6.001000,0.000000) (6.101000,0.000000) (6.201000,0.000000) (6.301000,0.000000) (6.401000,0.000000) (6.501000,0.000000) (6.601000,0.000000) (6.701000,0.000000) (6.801000,0.000000) (6.901000,0.000939) (7.001000,0.002160) (7.101000,0.002848) (7.201000,0.003352) (7.301000,0.003751) (7.401000,0.004078) (7.501000,0.004353) (7.601000,0.004586) (7.701000,0.004786) (7.801000,0.004958) (7.901000,0.005106) (8.001000,0.005234) (8.101000,0.005344) (8.201000,0.005438) (8.301000,0.005519) (8.401000,0.005587) (8.501000,0.005644) (8.601000,0.005691) (8.701000,0.005729) (8.801000,0.005758) (8.901000,0.005780) (9.001000,0.005794) (9.101000,0.005802) (9.201000,0.005804) (9.301000,0.005801) (9.401000,0.005792) (9.501000,0.005779) (9.601000,0.005760) (9.701000,0.005738) (9.801000,0.005711) (9.901000,0.005681) (10.001000,0.005647) (10.101000,0.005610) (10.201000,0.005570) (10.301000,0.005526) (10.401000,0.005479) (10.501000,0.005430) (10.601000,0.005378) (10.701000,0.005323) (10.801000,0.005266) (10.901000,0.005207) (11.001000,0.005145) (11.101000,0.005081) (11.201000,0.005014) (11.301000,0.004945) (11.401000,0.004874) (11.501000,0.004801) (11.601000,0.004726) (11.701000,0.004648) (11.801000,0.004569) (11.901000,0.004487) (12.001000,0.004403) (12.101000,0.004316) (12.201000,0.004228) (12.301000,0.004136) (12.401000,0.004043) (12.501000,0.003947) (12.601000,0.003848) (12.701000,0.003746) (12.801000,0.003642) (12.901000,0.003534) (13.001000,0.003423) (13.101000,0.003307) (13.201000,0.003188) (13.301000,0.003065) (13.401000,0.002936) (13.501000,0.002802) (13.601000,0.002661) (13.701000,0.002513) (13.801000,0.002356) (13.901000,0.002188) (14.001000,0.002007) (14.101000,0.001809) (14.201000,0.001587) (14.301000,0.001330) (14.401000,0.001011) (14.501000,0.000528) (14.601000,-0.000000) (14.701000,-0.000000) (14.801000,-0.000000) (14.901000,-0.000000)
      };
      \addplot[only marks,mark=x,line width=0.5pt] plot coordinates{
(0.725210,0)(0.761293,0)(0.773428,0)(0.804447,0)(0.850258,0)(0.876619,0)(0.914223,0)(0.936101,0)(0.966339,0)(0.983292,0)(1.000177,0)(1.037538,0)(1.049136,0)(1.078080,0)(1.126212,0)(1.159059,0)(1.193031,0)(1.232553,0)(1.295855,0)(1.357692,0)(2.087068,0)(2.322090,0)(2.404135,0)(2.428502,0)(2.523406,0)(2.591069,0)(2.611899,0)(2.760563,0)(2.842153,0)(2.931724,0)(3.005278,0)(3.114150,0)(3.205625,0)(3.330611,0)(3.448666,0)(3.493621,0)(3.634134,0)(3.752682,0)(3.923243,0)(4.018420,0)(7.376547,0)(7.714738,0)(7.892245,0)(8.144817,0)(8.710341,0)(8.864955,0)(9.216528,0)(9.416072,0)(9.976535,0)(10.291326,0)(10.398824,0)(10.794176,0)(11.077766,0)(11.360634,0)(11.786087,0)(12.331401,0)(12.456995,0)(12.879208,0)(13.498515,0)(13.942644,0)
      };
      \legend{{Asymptotic spectrum},{Empirical eigenvalues}}
    \end{axis}
  \end{tikzpicture}
  \end{center}
  \caption{Empirical and asymptotic eigenvalue distribution of $\hat{{\bf R}}_N$ for $L=3$, $\rho_1=1$, $\rho_2=3$, $\rho_3=10$, $N/M=c=0.1$, $N=60$, $N_1=N_2=N_3=20$.}
  \label{fig:clusters1}
\end{figure}
\end{minipage}\hfill
\begin{minipage}[c]{0.43\textwidth}
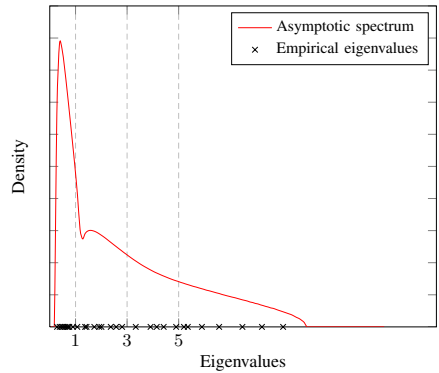
\begin{figure}[H]
  \begin{center}
  \begin{tikzpicture}[font=\footnotesize,scale=0.75]
    \renewcommand{\axisdefaulttryminticks}{8}
    \tikzstyle{every major grid}+=[style=densely dashed]
\pgfplotsset{every axis y label/.append style={yshift=-20pt}}
\pgfplotsset{every axis x label/.append style={yshift=5pt}}
    \tikzstyle{every axis legend}+=[cells={anchor=west},fill=white,
        at={(0.98,0.98)}, anchor=north east, font=\scriptsize ]

    \begin{axis}[
      xmajorgrids=true,
      xlabel={Eigenvalues},
      ylabel={Density},
      xtick = {1,3,5},
      ytick = {},
      yticklabels = {},
      xmin=0,
      xmax=15,
      ymin=0,
      ymax=0.2,
      ]

      \addplot[smooth,red,line width=0.5pt] plot coordinates{
	  (0.075128,0.000355)(0.175128,0.000071)(0.275128,0.132883)(0.375128,0.175839)(0.475128,0.173277)(0.575128,0.161351)(0.675128,0.146918)(0.775128,0.131656)(0.875128,0.116048)(0.975128,0.099919)(1.075128,0.082234)(1.175128,0.060802)(1.275128,0.054855)(1.375128,0.058257)(1.475128,0.059739)(1.575128,0.060151)(1.675128,0.059893)(1.775128,0.059315)(1.875128,0.058449)(1.975128,0.057396)(2.075128,0.056302)(2.175128,0.055113)(2.275128,0.053836)(2.375128,0.052578)(2.475128,0.051307)(2.575128,0.050069)(2.675128,0.048806)(2.775128,0.047592)(2.875128,0.046334)(2.975128,0.045169)(3.075128,0.043956)(3.175128,0.042812)(3.275128,0.041680)(3.375128,0.040624)(3.475128,0.039547)(3.575128,0.038535)(3.675128,0.037581)(3.775128,0.036687)(3.875128,0.035787)(3.975128,0.034931)(4.075128,0.034095)(4.175128,0.033378)(4.275128,0.032665)(4.375128,0.031955)(4.475128,0.031309)(4.575128,0.030643)(4.675128,0.030004)(4.775128,0.029406)(4.875128,0.028849)(4.975128,0.028360)(5.075128,0.027766)(5.175128,0.027284)(5.275128,0.026728)(5.375128,0.026250)(5.475128,0.025740)(5.575128,0.025232)(5.675128,0.024832)(5.775128,0.024340)(5.875128,0.023875)(5.975128,0.023431)(6.075128,0.022899)(6.175128,0.022534)(6.275128,0.022093)(6.375128,0.021636)(6.475128,0.021177)(6.575128,0.020783)(6.675128,0.020246)(6.775128,0.019822)(6.875128,0.019377)(6.975128,0.018982)(7.075128,0.018554)(7.175128,0.018207)(7.275128,0.017759)(7.375128,0.017301)(7.475128,0.016852)(7.575128,0.016480)(7.675128,0.016062)(7.775128,0.015484)(7.875128,0.015170)(7.975128,0.014623)(8.075128,0.014313)(8.175128,0.013729)(8.275128,0.013293)(8.375128,0.012778)(8.475128,0.012377)(8.575128,0.011821)(8.675128,0.011524)(8.775128,0.010933)(8.875128,0.010356)(8.975128,0.009823)(9.075128,0.009280)(9.175128,0.008727)(9.275128,0.008158)(9.375128,0.007597)(9.475128,0.006782)(9.575128,0.005900)(9.675128,0.005193)(9.775128,0.004054)(9.875128,0.002703)(9.975128,0.000000)(10.075128,0.000000)(10.175128,0.000000)(10.275128,0.000000)(10.375128,0.000000)(10.475128,0.000000)(10.575128,0.000000)(10.675128,0.000000)(10.775128,0.000000)(10.875128,0.000000)(10.975128,0.000000)(11.075128,0.000000)(11.175128,0.000000)(11.275128,0.000000)(11.375128,0.000000)(11.475128,0.000000)(11.575128,0.000000)(11.675128,0.000000)(11.775128,0.000000)(11.875128,0.000000)(11.975128,0.000000)(12.075128,0.000000)(12.175128,0.000000)(12.275128,0.000000)(12.375128,0.000000)(12.475128,0.000000)(12.575128,0.000000)(12.675128,0.000000)(12.775128,0.000000)(12.875128,0.000000)(12.975128,0.000000)
      };
      \addplot[only marks,mark=x,line width=0.5pt] plot coordinates{
(0.283395,0.000000)(0.384356,0.000000)(0.452983,0.000000)(0.514274,0.000000)(0.584669,0.000000)(0.660092,0.000000)(0.704346,0.000000)(0.756649,0.000000)(0.899373,0.000000)(1.068115,0.000000)(1.371647,0.000000)(1.427020,0.000000)(1.729349,0.000000)(1.921773,0.000000)(2.004953,0.000000)(2.366641,0.000000)(2.577564,0.000000)(2.807198,0.000000)(3.335924,0.000000)(3.910238,0.000000)(4.150556,0.000000)(4.419376,0.000000)(4.903128,0.000000)(5.206239,0.000000)(5.355599,0.000000)(5.902938,0.000000)(6.571948,0.000000)(7.474017,0.000000)(8.231752,0.000000)(9.052991,0.000000)      };
      \legend{{Asymptotic spectrum},{Empirical eigenvalues}}
    \end{axis}
  \end{tikzpicture}
  \end{center}
  \caption{Empirical and asymptotic eigenvalue distribution of $\hat{{\bf R}}_N$ for $L=3$, $\rho_1=1$, $\rho_2=3$, $\rho_3=5$, $N/M=c=3/8$, $N=30$, $N_1=N_2=N_3=10$.}
  \label{fig:clusters2}
\end{figure}
\end{minipage}
\vspace{0.5cm}

A way to circumvent the separability condition has recently been proposed by Bai, Chen and Yao \cite{chen}, based on the use of the empirical asymptotic moments:
$$
\hat{\alpha}_k=\frac{1}{M}\tr\left({\bf Y}_N {\bf Y}_N\right)^k, k\in\left\{1,\cdots,2L\right\},
$$
which can be shown to be a  sufficient statistics to estimate $\left(\frac{N_1}{N},\cdots,\frac{N_L}{N},\rho_1,\cdots,\rho_L\right)$. Although being robust to separability condition, this technique suffers from numerical difficulties, since the proposed estimator has no closed-form expression and thus should be determined numerically. An interesting contribution, although not directly focused on estimating the covariance of the observations is the work by Rubio and Mestre \cite{rubio09-SP}, where an alternative way to estimates the moments  
$$\gamma_k= \frac{1}{N}\mathrm{Tr} ({\bf R}_N^k),$$ for all ${k \in \mathbb N}$ is proposed, yielding an explicit (yet lengthy) formula. 

In this paper, we improve existing work in several directions: With
respect to Mestre's seminal papers \cite{MES08,MES08b}, we propose a
joint estimation of the eigenvalues and their multiplicities, and drop
the separability assumption. The proposed estimator is close in spirit
to the one in \cite{chen}, although we carefully establish the
existence and uniqueness of the estimator, a fact that is not explicit
in \cite{chen} (we shall also mention a close ongoing work by Li and
Yao, not yet disclosed to our knowledge) . Finally, we study the
fluctuations of the estimator and establish a Central Limit theorem.

The remainder of the paper is organized as follows. In Section
\ref{sec:mes}, the main assumptions are provided and Mestre's
estimator \cite{MES08b} is briefly reviewed. In Section
\ref{sec:estimation}, the proposed estimator is described. Its
fluctuations are studied in Section \ref{sec:fluctuations}, where a
central limit theorem is stated. Simulations are presented in Section
\ref{sec:simulations}, and a discussion ends the paper in Section
\ref{sec:discuss}. Finally, the remaining technical details are
provided in the Appendix.

\section{Main assumptions and general background}
\label{sec:mes}

\subsection{Notations}
In this paper, the notations $s,{\bf x}, {\bf M}$ stand for scalars,
vectors and matrices, respectively. Superscripts $(\cdot)^T$ and
$(\cdot)^H$ respectively stand for the transpose and transpose
conjugate; trace of ${\bf M}$ is denoted by $\mathrm{Tr}( {\bf
  M})$; determinant of ${\bf M}$, by $\mathrm{det}({\bf
  M})$; the mathematical expectation operator, by $\mathbb{E}$. If
$z\in \mathbb{C}$, then $\Re(z)$ and $\Im(z)$ respectively stand for
$z$'s real and imaginary parts, while $\mathbf{i}$ stands for
$\sqrt{-1}$; $\overline{z}$ stands for $z$'s conjugate.

If ${\bf Z} \in \mathbb{C}^{N\times N}$ is a nonnegative Hermitian matrix
with eigenvalues $(\xi_i; \ 1\le i\le N)$, we denote in the sequel by $F^{\bf Z}$
the empirical distribution of its eigenvalues (also called
{\em spectral distribution} of $\bf Z$), {\em i.e.}:
$$
F^{\bf Z}(d\,\lambda) = \frac 1N \sum_{i=1}^N \delta_{\xi_i}(d\,\lambda)\ ,
$$
where $\delta_x$ stands for the Dirac probability measure at $x$.

Convergence in distribution will be denoted by
$\xrightarrow[]{\mathcal D}$, in probability by
$\xrightarrow[]{\mathcal P}$; and almost sure convergence, by
$\xrightarrow[]{a.s.}$.

\subsection{Main assumptions}

Consider the model
$$
{\bf Y}_N = {\bf R}_N^{1/2} {\bf X}_N,\
$$
and

$$\hat{\bf R}_N = \frac{1}{M} {\bf Y}_N {\bf Y}_N ^{H}.$$

At first, an assumption about the matrix ${\bf R}_N$ is needed: 
\begin{assumption} \label{ass:RN}

${\bf R}_N$ is a $N\times N$ Hermitian non-negative definite
matrix with $L$ ($L$ being fixed) distinct eigenvalues $0<\rho_1<\cdots<
\rho_L$ with respective multiplicities $N_1,\cdots, N_L$ (notice that
$\sum_{i=1}^L N_i = N$).

\end{assumption}

As mentioned earlier, we consider the asymptotic regime where the
number of samples $M$ and the number of variables $N$ grow to infinity
at the same pace, together with the multiplicities of each of ${\bf
  R}_N$'s eigenvalues.

\begin{assumption} \label{ass:NM}    Let $M,N$ be integers such that: \begin{equation} N,M\to \infty\ ,\quad   \textrm{with}\quad  \frac NM \to c\in (0,\infty)\ ,\quad
\textrm{and}\quad \frac {N_i}N \to c_i\in (0,\infty)\ , \ 1\le i\le L. \end{equation}
This assumption will be shortly referred to as $N,M\to \infty$.
\end{assumption}

The following assumption is standard and is sufficient for estimation
purposes.

\begin{assumption} \label{ass:entries}

Let ${\bf X}_N=(X_{ij})$ be a $N\times M$ matrix whose entries are
i.i.d. random variables in $\mathbb C$ such that $\mathbb{E}({\bf
  X}_{1,1})=0$, $\mathbb{E}(|{\bf X}_{1,1}|^2)=1$ with finite fourth
moment: $\mathbb{E}(|{\bf X}_{1,1}|^4) < \infty$.
\end{assumption}

\begin{remark}
In order to establish the fluctuations of this estimator, the
Gaussianity of the entries of ${\bf X}_N$ is needed (although this
technical condition may be removed with substantial extra work).
\end{remark}

{\bf Assumption 3b}: The entries of the $N\times M$ matrix ${\bf X}_N=(X_{ij})$ are i.i.d. standard complex Gaussian variables, {\it i.e.}  $X_{ij} = U +
\mathbf{i} V$, where $U,V$ are both independent real Gaussian random variables
${\mathcal N}(0,\frac{1}{2})$.

It is well-known in large random matrix theory that under Assumptions
\ref{ass:RN}, \ref{ass:NM} and \ref{ass:entries}, $F^{\hat{\bf
    R}_N}$ converges to a limiting probability distribution. In
Mestre's paper \cite{MES08b}, a {\em separability
  condition}\footnote{The precise technical statement of the
  separability condition together with a mathematical interpretation
  are available in \cite{MES08b}, but are not necessary here.} is
needed in order to derive the estimator of ${\bf R}_N$'s eigenvalues:

\begin{assumption}
\label{ass:separability}
The support $\mathcal S$ of the limiting probability
distribution of $F^{\hat{\bf R}_N}$ is composed of $L$ compact connex disjoint
subsets, and not reduced to a singleton. 
\end{assumption}
\begin{remark} Note that when $M<N$, matrix $\hat{\bf R}_N$ is singular and thus admits $(N-M)$ eigenvalues equal to zero. Hence, the limiting spectrum of   $\hat{\bf R}_N$
 has an additional mass in zero with weight $1-\frac{1}{c}$, which will not be considered among the $L$ clusters.
\end{remark}

The separability condition is illustrated in Fig. \ref{fig:clusters1}
and \ref{fig:clusters2}. In both figures, the limiting distribution of
$F^{\hat{\bf R}_N}$ is drawn (red line). In Fig. \ref{fig:clusters1},
${\bf R}_N$'s eigenvalues are $\rho_1=1$, $\rho_2=3$, $\rho_3=10$,
they have the same multiplicity and the ratio $c$ is equal to
$0.1$. In this case, the separability condition is satisfied as the
limiting distribution exhibits 3 clusters. The separability condition
is no longer satisfied in Fig. \ref{fig:clusters2}, where $\rho_1=1$,
$\rho_2=3$, $\rho_3=5$ and $c=0.375$, but where the limiting
distribution only exhibits a single cluster.
 
\subsection{Background on Large Random Matrices, Mestre's estimators
  and their fluctuations}

The {\it Stieltjes transform} has proved since Mar\v{c}enko and
Pastur's seminal paper \cite{MAR67} to be extremely efficient to
describe the limiting spectrum of large dimensional random matrices.
Given a probability distribution $\mathbb{P}$ defined over
$\mathbb{R}^+$, its Stieltjes transform is a $\mathbb{C}$-valued
function defined by:
$$
m_{\mathbb{P}}(z)= \int_{\mathbb{R}^+}
\frac{\mathbb{P}(d\lambda)}{\lambda-z}\ , \quad z\in \mathbb{C}
\backslash \mathbb{R}^+\ .
$$
In the case where $F^{\bf Z}$ is the spectral distribution associated
to a nonnegative Hermitian matrix ${\bf Z}\in\mathbb{C}^{N\times N}$
with eigenvalues $(\xi_i; \ 1\le i\le N)$, the Stieltjes transform
$m_{\bf Z}$ of $F^{\bf Z}$ takes the particular form:
\begin{eqnarray*}
  m_{\bf Z}(z) &=&\int \frac{F^{\bf Z}(d\,\lambda)} {\lambda-z}  \\
&=& \frac1N\sum_{i=1}^N \frac1{\xi_i-z} \ = \ \frac 1N \tr \left( {\bf Z} -z{\bf I}_N\right)^{-1}\ ,
\end{eqnarray*}
which is exactly the normalized trace of the resolvent $\left( {\bf Z} -z{\bf I}_N\right)^{-1}$.

An important result associated to the model under investigation here
is Bai and Silverstein's description of the limiting spectral
distribution of $\hat{\bf R}_N$ \cite{SIL95} (see also \cite{MAR67}):

\begin{theorem} \cite{SIL95}
\label{th:stieltjes}
Assume that Assumptions \ref{ass:RN}, \ref{ass:NM}, \ref{ass:entries} hold true and denote by
$F^{\bf R}$ the limiting spectral distribution of ${\bf R}_N$, {\em i.e.}
$F^{\bf R}(d\,\lambda) = \sum_{k=1}^L c_k \delta_{\rho_k}(d\,\lambda)$. The spectral distribution $F^{\hat{\bf
    R}_N}$ of the sample covariance matrix $\hat{\bf R}_N$ converges
(weakly and almost surely) to a probability distribution $F$ as
$M,N\to \infty$, whose Stieltjes transform $m(z)$ satisfies:
$$
m(z)=\frac1c \underline{m}(z) -
\left(1-\frac1c\right)\frac1z\ ,
$$
for $z\in \mathbb{C}^+ = \{ z\in \mathbb{C},\ \Im(z)>0\}$, where $\underline{m}(z)$ is
defined as the unique solution in $\mathbb{C}^+$ of:
$$
\underline{m}(z) = - \left( z - c\int \frac{t}{1+t\underline{m}(z)}dF^{\bf R}(t) \right)^{-1}.
$$

\end{theorem}

\begin{remark}
Note that $\underline{m}(z)$ is also a Stieltjes transform whose
associated distribution function will be denoted $\underline{F}$,
which turns out to be the limiting spectral distribution of
$F^{\hat{\underline{\bf R}}_N}$ where $\hat{\underline{\bf R}}_N$ is
defined as:
$$
\hat{\underline{\bf R}}_N \triangleq \frac{1}{M}{\bf X}_N^H {\bf R}_N{\bf X}_N\ .
$$
\end{remark}

\begin{remark}
Denote by $m_{ \hat{\bf R}_N}(z)$ and $m_{ \hat{\underline{\bf
      R}}_N}(z)$ the Stieltjes transforms of $F^{ \hat{\bf R}_N}$ and
$F^{\hat{\underline{\bf R}}_N}$. Notice in particular that
\begin{equation}\label{eq:def-Mhat}
m_{ \hat{\bf R}_N}(z) = \frac{M}N m_{ \hat{\underline{\bf
      R}}_N}(z) - \left(1 - \frac{M}N\right)\frac1z\ .
\end{equation}
\end{remark}

\begin{remark}
Denote by $m_N(z)$ and $\underline{m}_N(z)$
the finite-dimensional counterparts of $m(z)$ and $\underline{m}(z)$,
respectively, defined by the relations:
\begin{equation}\label{eq:equivalent}
  \begin{cases}
  \underline{m}_N(z) = - \left( z - \frac{N}M\int \frac{t}{1+t\underline{m}_N(z)}dF^{{\bf R}_N}(t) \right)^{-1}\ ,& \text{}\\
  m_N(z) = \frac{M}N \underline{m}_N(z) - \left(1 - \frac{M}N\right)\frac1z\ .& \text{}\\
\end{cases}
\end{equation}
It can be shown that $m_N$ and $\underline{m}_N$ are Stieltjes transforms of given probability measures
$F_N$ and $\underline{F}_N$, respectively (cf. \cite[Theorem 3.2]{Coubook}).
\end{remark}

In \cite{MES08b}, Mestre proposes a novel approach to estimate the
eigenvalues $(\rho_k; \ 1\le k\le L)$ of the population covariance
matrix based on the observations $\hat{\bf R}_N$ under the additional Assumption \ref{ass:separability}. His approach relies on large random matrix theory and
the separability condition presented above plays a major role in the
mere definition of the estimators. As it will be a useful background 
in the sequel, we provide hereafter a brief description of Mestre's results:

\begin{theorem}\cite{MES08b} \label{th:mestre}
  \label{prop:mestre} Denote by
  $\hat{\lambda}_1 \leq \cdots \leq \hat{\lambda}_N$ the ordered eigenvalues of $\hat{\bf R}_N$.
   Under Assumptions \ref{ass:RN}, \ref{ass:NM}, \ref{ass:entries}, \ref{ass:separability} and assuming moreover that the multiplicities $N_1,\cdots,N_L$ are known, the following convergence holds true:
  \begin{equation} \label{eq:estimator}
    \tilde{\rho}_k - \rho_k \xrightarrow[M,N\to \infty]{a.s.} 0\ ,
  \end{equation}
  where
  \begin{equation}
          \label{eq:Mestre_tk}
	  \tilde{\rho}_k = \frac{M}{N_k}\sum_{m\in \mathcal N_k}\left(\hat\lambda_m - \hat{\mu}_m\right)\ ,
  \end{equation}
  with $\mathcal N_k =\{ \sum_{j=1}^{k-1}N_j +1,\ldots,\sum_{j=1}^kN_j
  \}$ and $\hat\mu_1 \leq \cdots\leq \hat\mu_N$ the (real and) ordered
  solutions of:
\begin{equation}\label{eq:def-mu}
\frac 1N \sum_{m=1}^N \frac{\hat \lambda_m}{\hat\lambda_m - \mu} = \frac MN\ 
\end{equation}
repeated with their multiplicites. When $N > M$, we use the convention $\hat{\mu}_1=\cdots=\hat{\mu}_{N-M+1}=0,$ whereas $\hat{\mu}_{N-M+2},\cdots, \hat{\mu}_N$ contain the positive solutions to the above equation. 


\end{theorem}

\begin{remark} Notice that \eqref{eq:def-mu} associated to
  \eqref{eq:def-Mhat} readily implies that for non null $\hat{\mu}_i$, $m_{ \hat{\underline{\bf
        R}}_N}(\hat \mu_i) = 0$. Otherwise stated, the $\hat\mu_i$'s
  are the zeros of $m_{ \hat{\underline{\bf R}}_N}$. This fact will be
  of importance in the sequel.
\end{remark}


\begin{proofsketch} We can now describe the main steps of Theorem \ref{th:mestre}. By Cauchy's formula, write:
$$
\rho_k = \frac{N}{N_k} \frac{1}{2 i \pi} \oint_{\Gamma_k} \left(
  \frac{1}{N}\sum_{r=1}^{L} N_r \frac{w}{\rho_r-w} dw \right)\ ,
$$
where $\Gamma_k$ is a positively oriented (clockwise) contour taking values on
$\mathbb{C} \setminus\{\rho_1,\cdots, \rho_L\}$ and only enclosing
$\rho_k$. With the change of variable $w=-\frac{1}{\underline{m}_{M}(z)}$
and the condition that the limiting support $\mathcal S$ of the
eigenvalue distribution of ${\bf R}_N$ is formed of $L$ distinct clusters
$({\mathcal S}_k, 1\le k\le L)$ (cf. Figure \ref{fig:clusters1}), we can write:
\begin{equation}
\label{eq:eigenvalue}
  \rho_k = \frac{M}{2 i \pi N_k}\oint_{\mathcal C_k} z\frac{\underline{m}_N'(z)}
{\underline{m}_N(z)}dz\ ,\quad 1\le k\le L ,\
\end{equation}
where ${\mathcal C}_k$ denotes positively
oriented contours which enclose the corresponding clusters ${\mathcal S_k}$. Defining
\begin{equation}
	\label{eq:hattk}
        \tilde{\rho}_k \triangleq \frac{M}{2\pi i N_k}\oint_{\mathcal C_k} z\frac{m_{\hat{\underline{\bf R}}_N}'(z)}{m_{\hat{\underline{\bf R}}_N}(z)}dz\ ,\quad 1\le k\le L\ ,
\end{equation}
dominated convergence arguments ensure that $\tilde{\rho}_k-\rho_k\to
0$, almost surely. The integral form of $\tilde{\rho}_k$ can then be
explicitly computed thanks to residue calculus, and this finally
yields \eqref{eq:Mestre_tk}.
\end{proofsketch}

Recently a central limit theorem has been derived \cite{YAO11journal} for this estimator under the extra 
assumption that the entries of ${\bf X}_N$ are Gaussian .
\begin{theorem} \cite{YAO11journal}\label{th:CLT} With the same notations as before, under Assumptions \ref{ass:RN}, \ref{ass:NM}, 3b, \ref{ass:separability} and with known multiplicities $N_1, \cdots,N_L$, then: 
$$
\left( M(\tilde{\rho}_k- \rho_k), \ 1\le k\le L \right)
 \xrightarrow[M,N\to \infty]{\mathcal D} {\bf x} \sim {\mathcal N}_L(0,\boldsymbol{\Theta})\ ,
$$
where ${\mathcal N}_L$ refers to a real $L$-dimensional
Gaussian distribution, and $\boldsymbol{\Theta}$ is a $L\times L$
matrix whose entries $\Theta_{k\ell}$ are given by,

$$\Theta_{k\ell} \quad = \quad -\frac1{4\pi^2 c^2c_{k}c_{\ell}}\oint_{\mathcal C_{k}}\oint_{\mathcal C_{\ell}}
\left[ \frac{\underline{m}'(z_1)\underline{m}'(z_2)}{(\underline{m}(z_1)-\underline{m}(z_2))^2}
-\frac1{(z_1-z_2)^2} \right] \frac1{\underline{m}(z_1)\underline{m}(z_2)}dz_1 dz_2\ ,$$
where ${\mathcal C}_k$ (resp. ${\mathcal C}_\ell$) is a closed counterclockwise oriented contour which only contains the k-th cluster (resp. $\ell$-th) .
\end{theorem}

The proof of this theorem is based on \cite{BAI04} and the continuous
mapping theorem. Details are available in \cite{YAO11journal}.

The main objective of this article is to provide estimators for the
$\rho_k$'s without relying anymore on the separability condition
(i.e. to remove Assumption \ref{ass:separability}). A Central Limit Theorem will
be established as well for the proposed estimator. As a by-product,
the knowledge of the multiplicities will no longer be needed, and they
will be estimated as well.

\section{Estimation of the eigenvalues $\rho_i$}
\label{sec:estimation}

In this section, we provide a method to estimate consistently the eigenvalues of the population covariance matrix without the  need to the separability condition (cf. Fig. \ref{fig:clusters2}). Our method is based on the asymptotic evaluation of the moments of the  eigenvalues of ${\bf R}_N$, ${\gamma}_i=\sum_{k=1}^L \frac{N_k}{N}\rho_k^i$, $1\le i \le 2L-1$. If $(\widehat{m}_i)_{1\le i \le 2L-1}$ are the empirical moments of the sample eigenvalues, then it is well known that except for $i=1$, ${\gamma}_i$ cannot be approximated by  $\widehat{m}_i$. Consistent estimators for $\gamma_i$ are provided in \cite{rubio09-SP}, where it has been proved that:
$$
\gamma_i-\tilde{\gamma}_i\xrightarrow[N,M\to+\infty]{}0,
$$
where
$$
\tilde{\gamma}_i=\sum_{l=1}^i\mu_S(l,i) \widehat{m}_l,
$$
$\mu_S(l,i)$ being some given coefficients that depend on the system dimensions and on the empirical moments $\widehat{m}_i$ \cite{rubio09-SP}. An alternative is to use the Stieltjes transform:
\begin{lemma}
\label{lem:moments}
Assume that Assumptions \ref{ass:RN}, \ref{ass:NM} and \ref{ass:entries} hold true.
Let $\hat{\gamma}_i$ be the real quantities given by:
$$
\left\{
\begin{array}{lll}
\hat{\gamma}_0&=&1,\\
\hat{\gamma}_1&=&-\frac{M}{2N{\bf i}\pi}\oint_{\mathcal{C}}\frac{z\mrnhat'(z)}{\mrnhat(z)}dz,\\
\hat{\gamma}_{k}&=&\frac{M(-1)^{k}}{2Nk{\bf i}\pi}\oint_{\mathcal{C}}\frac{dz}{\mrnhat^{k}(z)}, \hspace{0.5cm} \textnormal{for} \hspace{0.2cm} 2\leq k\leq 2L-1
\end{array}
\right.
$$
\label{lem:moment}
where $\mathcal{C}$ is a counterclockwise oriented contour which encloses the support $\mathcal{S}$ of the limiting distribution of the eigenvalues of $\rnhat$. Let $\gamma_i$ be the moments of the eigenvalues of ${\bf R}_N$, {\it i.e.} ${\gamma}_i=\sum_{k=1}^L \frac{N_k}{N} \rho_k^i$. Then, for $1\le i\le 2L-1$,
$$
\hat{\gamma}_i-\gamma_i\xrightarrow[N,M\to\infty]{a.s.} 0 \ .
$$
\end{lemma}
The proof of this lemma is postponed to Appendix \ref{app:proof_lemma}. While the estimates proposed by \cite{rubio09-SP}
are better in practice, estimates $(\hat\gamma_i)$ will be of interest in order to establish the central limit theorem, and  
to obtain a closed-form expression of the asymptotic variance.

An interesting remark is that the map that links the eigenvalues and their multiplicities to their first $2L-1$ moments is invertible. Retrieving the eigenvalues from the estimates of the  $2L-1$ moments is thus possible. This is the basic idea on which our method is founded.




The main result is stated as below:
\begin{theorem}
\label{th:estimation}
Recall the notations of Lemma \ref{lem:moments} and consider the system of equations:
\begin{equation}
\label{eq:estimations}
\begin{cases}
\sum_{i=1}^L x_i = 1 ,& \text{}\\

\sum_{i=1}^L x_i y_i= \hat{\gamma}_1, &\text{}\\

\sum_{i=1}^L x_i y_i^k=\hat{\gamma}_k & \text{for $2 \leq k \leq 2L-1$,}
\end{cases}
\end{equation}
where $(x_i)_{1\leq i\leq L}$ and $(y_i)_{1 \leq i \leq L}$ are $2L$ unknown parameters.
Then under Assumptions \ref{ass:RN}, \ref{ass:NM}, \ref{ass:entries}, the system of equations (\ref{eq:estimations}) has one and only one real solution $(\hat{c}_1,\cdots,\hat{c}_L, \hat{\rho}_1,\cdots,\hat{\rho}_L)$ with $\hat{\rho}_1 \leq \cdots \leq \hat{\rho}_L.$ Moreover, $(\hat{c}_1,\cdots,\hat{c}_L,\hat{\rho}_1,\cdots,\hat{\rho}_L)$ is a consistent estimator of $(c_1,\cdots,c_L,\rho_1,\cdots,\rho_L)$, {\it{i.e.}},
$$\hat{c}_\ell- c_\ell \xrightarrow[N,M\rightarrow \infty]{a.s.} 0\quad \textrm{and}\quad  \hat{\rho}_\ell-\rho_\ell \xrightarrow[N,M\rightarrow \infty]{a.s.} 0,$$
with $c_\ell= \lim\frac{N_\ell}{N}$ for $1 \leq \ell \leq L.$
\end{theorem}

\begin{remark}
The condition of separability is not required in the previous theorem. Moreover, the multiplicities are assumed to be unknown and thus have to be estimated. Fig \ref{fig:clusters2} represents a case where the three clusters are merged into one cluster. In such a situation, the estimator in \cite{MES08b} is biased whereas the proposed one is asymptotically consistent.

\end{remark}

\begin{remark}
We use the estimator proposed in Lemma \ref{lem:moments}. However, the proof below does not depend on the estimator of the moments we choose. In fact, for any consistent estimator of the moments $\gamma_i$, the above theorem always holds true.

\end{remark}
\begin{proof}
The proof can be split into two main steps. By using the inverse function theorem, we can prove the almost sure existence of a real solution. Then, the uniqueness is ensured by a matrix inversion argument.


\begin{enumerate}
\item {\it Existence of a real solution of the system.}
\label{sec:existence}
\end{enumerate}

The first task is to show that the system of equations \eqref{eq:estimations} admits, for $N$ sufficiently large, one real solution $(\hat{c}_1,\cdots,\hat{c}_L,\hat{\rho}_1\cdots,\hat{\rho}_L)$ satisfying $\hat{\rho}_1<\hat{\rho}_2<\cdots<\hat{\rho}_L$.  We shall also establish the consistency of the obtained solution. The proof of the existence of a real solution follows in the same way as in \cite{CHE10}. It is merely based on the use of the inverse function theorem which ensures the existence as soon as the Jacobian matrix of the considered transformation is invertible.
We recall below the inverse function theorem \cite{krantz}:
\begin{theorem} \cite{krantz}
Let $f$ : ${\mathbb R}^{n} \to {\mathbb R}^n$ be a continuously differentiable function. Let ${\bf a}$ and ${\bf b}$ be vectors of $\mathbb{R}^n$ such that $f({\bf a})={\bf b}$. If the Jacobian of $f$ at ${\bf a}$ is invertible, then there exists a neighborhood $U$ containing ${\bf a}$ such that $f:U\rightarrow f(U)$ is a diffeomorphism, i.e, for every ${\bf y}\in f(U)$ there exists a unique ${\bf x}$ such that $f({\bf x})={\bf y}$. In particular, $f$ is invertible in $U$.
\end{theorem}
Consider the functional $f$ defined as:
$$
f(x_1,\cdots,x_L,y_1,\cdots,y_L)=
\left(\sum_{\ell=1}^L x_{\ell}\ ,\ \sum_{\ell=1}^L x_{\ell}y_\ell\ ,\  \cdots\ ,\ \sum_{\ell=1}^L x_{\ell}y_{\ell}^{2L-1}\right).$$
Consider ${\bf z}=(x_1,\cdots,x_L,y_1,\cdots,y_L)$ and denote by ${\bf c}= (c_1,\cdots,c_L,\rho_1,\cdots,\rho_L)$; 
we then have:
$$
{\bf M}\triangleq \left.\frac{\partial f}{\partial {\bf z}}\right|_{{\bf z}={\bf c}}=\begin{bmatrix}
1 &\cdots & 1 &0 &\cdots &0 \\
\rho_1 & \cdots &\rho_L & c_1 & \cdots & c_L\\
\vdots & \ddots & \ddots & \ddots & \ddots & \vdots \\
\rho_1^{2L-1} &\cdots & \rho_L^{2L-1} & (2L-1)c_1\rho_1^{2L-2}& \cdots & (2L-1)c_L\rho_L^{2L-2}
\end{bmatrix}.
$$
We will show that ${\bf M}$ is invertible by contradiction. Assume that ${\bf M}$ is singular. Then, there exists a non null vector $\boldsymbol{\lambda}=\left[\lambda_1,\cdots,\lambda_{2L}\right]^{T}$ such that ${\bf M}^{T} \boldsymbol{\lambda}={\bf 0}$. Consider the polynomial
$$
\mathrm{P}(\mathrm{X})=\sum_{i=0}^{2L-1}\lambda_{i+1}\mathrm{X}^i.
$$
We easily observe that ${\bf M}^{T} \boldsymbol{\lambda}={\bf 0}$ implies that 
$$
\mathrm{P}(\rho_\ell)= \mathrm{P}'(\rho_\ell) = 0\ ,\quad \textrm{for}\ 1\le \ell \le L\ .
$$
In particular, the multiplicity of each $\rho_\ell$ is at least $2$. This is impossible since the degree of $\mathrm{P}$ is at most $2L-1$ (recall that all the eigenvalues $\rho_\ell$ are pairwise distinct). Matrix ${\bf M}$ is therefore invertible. The inverse function theorem then applies.
Denote by $\psi_i=\sum_{k=1}^L c_k \rho_k^i$ for $0\le i \le 2L-1$. There exists a neighborhood $U$ of $(c_1,\cdots,c_L,\rho_1,\cdots,\rho_L)$ and a neighborhood $V$ of $(\psi_0,\cdots,\psi_{2L-1})$ such that $f$ is a diffeomorphism from $U$ onto $V$.
On the other hand, we have:
$$
\hat{\gamma}_i-\gamma_i \xrightarrow[]{a.s.} 0.
$$
As $\gamma_i-\psi_i \to 0$, therefore, almost surely, $(\hat{\gamma}_0,\cdots,\hat{\gamma}_{2L-1})\in V$ for $N$ and $M$ large enough.
Hence, a real solution $$(\hat{c}_1,\cdots,\hat{c}_L,\hat{\rho}_1,\cdots,\hat{\rho}_L)=f^{-1}(\hat{\gamma}_0,\cdots,\hat{\gamma}_{2L-1})\in U$$ exists. And by the continuity, one can get easily that:
$$
\hat{c}_\ell-c_\ell \xrightarrow[N,M\rightarrow\infty]{a.s.} 0 \quad \textrm{and}\quad 
\hat{\rho}_\ell-\rho_\ell \xrightarrow[N,M\rightarrow\infty]{a.s.} 0 \quad \textrm{for} \quad   1\le \ell \le L \ .
$$

\begin{enumerate}
\item[2)]
{\it Uniqueness of the solution of the system.}
\end{enumerate}
Consider the polynomial $Q$ with degree $L$ defined as:
$$
\mathrm{Q}(\mathrm{X})=\prod_{\ell=0}^L (X-\hat{\rho}_\ell)\stackrel{\triangle}{=}\sum_{\ell=0}^L s_\ell \mathrm{X}^\ell
$$
where $s_L=1$. Denote by ${\bf s}={\left[s_0,\cdots,s_{L-1}\right]}^{T}$. It is clear that $g:\left(\hat{\rho}_1,\cdots,\hat{\rho}_L\right)\rightarrow {\bf s}$ is a homeomorphism. It remains thus to show that vector ${\bf s}$ is uniquely determined by $ (\hat{\gamma}_0,\cdots,\hat{\gamma}_{2L-1})$.

It is clear that each $\hat{\rho}_k$ is also the zero of the polynomial functions $\mathrm{R}_{\ell}(X)$ given by:
$$
\mathrm{R}_{\ell}(\mathrm{X})=\sum_{i=0}^L s_i \mathrm{X}^{i+\ell}\ ,
$$
where $0\le \ell \le L-1$. In other words, for $1\le k\le L$, we get:
$$
\sum_{i=0}^L s_i\hat{\rho}_k^{\ell+i}=0,
$$
or equivalently:
\begin{equation}
\sum_{i=0}^L s_i \hat{c}_k\hat{\rho}_k^{\ell+i}=0.
\label{eq:nonsum}
\end{equation}
Summing \eqref{eq:nonsum} over $k$, we obtain:
\begin{equation}
\sum_{i=0}^L \hat{\gamma}_{i+\ell}s_i=0\ , 
\label{eq:withsl}
\end{equation}
for $0\le \ell\le L-1$. Since $s_L=1$,  \eqref{eq:withsl} becomes:
\begin{equation}
 \hat{\gamma}_{L+\ell}+\sum_{i=0}^{L-1}s_i\hat{\gamma}_{i+\ell}=0\ ,
 \label{eq:nonmatricial}
\end{equation}
for $0\le \ell\le L-1$. 

Writing \eqref{eq:nonmatricial} in a matrix form, we get:
$
\boldsymbol{\Gamma}{\bf s}=-{\bf b},
$
where
$$
\boldsymbol{\Gamma}=\begin{bmatrix}
\hat{\gamma}_0 &\hat{\gamma}_1 &\cdots & \hat{\gamma}_{L-1}\\
\hat{\gamma}_1 & \hat{\gamma}_2 & \cdots &  \hat{\gamma}_{L}\\
\vdots & \ddots & \ddots & \vdots \\
\hat{\gamma}_{L-1} & \hat{\gamma}_{L} & \cdots & \hat{\gamma}_{2L-2}
\end{bmatrix} \qquad \textnormal{and}\qquad {\bf b}=\begin{bmatrix}
\hat{\gamma}_L\\
\vdots \\
\hat{\gamma}_{2L-1}
\end{bmatrix}.
$$

On the other hand, we have ${\bf \Gamma}={\bf A} {\bf D} {\bf A}^T$, where ${\bf D}={\rm diag}(\hat{c}_1, \hat{c}_2,\cdots,\hat{c}_L)$ and
$$
{\bf A}=\begin{bmatrix}
1 & 1  & \cdots & 1 \\
\hat{\rho}_1 & \hat{\rho}_2 & \cdots & \hat{\rho}_L \\
\vdots & \vdots &  & \vdots \\
\hat{\rho}_1^{L-1} & \hat{\rho}_2^{L-1} & \cdots &\hat{\rho}_L^{L-1}
\end{bmatrix}.
$$
 Then, $${\rm det}(\boldsymbol{\Gamma})=\prod_{k=1}^L \hat{c}_k\prod_{1\leq i< j\leq L}( \hat{\rho}_i- \hat{\rho}_j)^2>0.$$
Therefore, the vector ${\bf s}$ is then uniquely determined by $\boldsymbol{\Gamma}$ and  ${\bf b}$ and is given by:
$$
{\bf s}=-\boldsymbol{\Gamma}^{-1}{\bf b}.
$$
Hence the unicity. The proof is complete.

\end{proof}

\section{Fluctuations of the estimator}
\label{sec:fluctuations}
In this section, we shall study the fluctuations of the multiplicities and eigenvalues estimators  $(\hat{c}_1,\cdots,\hat{c}_L,\hat{\rho}_1,\cdots,\hat{\rho}_L)$ introduced in Theorem \ref{th:estimation}. In particular, we establish a central limit theorem for the whole vector in the case where the entries of matrix ${\bf X}_N$ are Gaussian.


\begin{theorem}

  Let Assumptions \ref{ass:RN}, \ref{ass:NM}, 3b hold true.  Let
  $(\hat{c}_1,\cdots,\hat{c}_L,\hat{\rho}_1,\cdots,\hat{\rho}_L)$ be
  the estimators obtained in Theorem \ref{th:estimation}. Then

$$M\left[\hat{c}_1-\frac{N_1}{N},\cdots, \hat{c}_L-\frac{N_L}{N}, \hat{\rho}_1-\rho_1,\cdots,\hat{\rho}_L-\rho_L\right] \xrightarrow[N,M\rightarrow \infty]{\mathcal D} \mathcal N_{2L}(0, \boldsymbol{\Theta}) $$
where $\boldsymbol{\Theta}$ is a $2L\times 2L$ matrix admitting the
decomposition $\boldsymbol{\Theta}={\bf M}^{-1} {\bf W} {{\bf
    M}^{-1}}^T$ with
$${\bf M}=
\left[\begin{matrix}
1 & \cdots & 1&0& \cdots & 0\\
\rho_1 & \cdots & \rho_L & c_1 & \cdots & c_L \\
\vdots&\ddots&\ddots & \ddots&\ddots& \vdots\\
\rho_1^{2L-1}&\cdots & \rho_L^{2L-1}&(2L-1) c_1 \rho_1^{2L-2} & \cdots & (2L-1) c_L \rho_L^{2L-2}\\
\end{matrix} \right]$$ and
$$
{\bf W}=\left[
\begin{BMAT}{cc}{cc}
0&{\bf 0}\\
{\bf 0} & {\bf V}
\end{BMAT}\right].
$$

where ${\bf V}$ is a $(2L-1) \times (2L-1)$ matrix whose entries are
given by (for $1 \leq k,\ell \leq 2L-1$):
$${V}_{k,\ell}=  -\frac{(-1)^{k+\ell}}{4 \pi^2 c^2}
\oint_{\mathcal C_1} \oint_{\mathcal C_2} \left(\frac{\underline{m}'(z_1) \underline{m}'(z_2)} {(\underline{m}(z_1)- \underline{m}(z_2))^2}-\frac{1}{(z_1-z_2)^2}\right) \times \frac{1}{\underline{m}^{k}(z_1) \underline{m}^{\ell}(z_2)} d\, z_1 d\, z_2$$
where $\mathcal{C}_1$ and $\mathcal{C}_2$ are two closed contours non-overlapping which contain the support $\mathcal S$ of $F$ and are counterclockwise oriented.

\end{theorem}

%
%
%
\begin{proof}
The proof relies on the same techniques used in \cite{YAO11journal}. We outline hereafter the main steps and provide then the details.

By Theorem \ref{th:estimation}, the estimate vector $\left(\hat{c}_1,\cdots,\hat{c}_L,\hat{\rho}_1,\cdots,\hat{\rho}_L\right)$ verifies the following system of equations:
$$
\left\{\begin{array}{l}
\sum_{i=1}^L \hat{c}_i=1,\\
\sum_{i=1}^L \hat{c}_i \hat{\rho}_i=\hat{\gamma}_1,\\
\sum_{i=1}^L \hat{c}_i \hat{\rho}_i^k=\hat{\gamma}_k
\ \textrm{for}\ 2\leq k\leq 2L-1,
\end{array}\right.
$$
where the $\hat{\gamma}_i$'s are the moment estimates provided by Lemma \ref{lem:moment}.


Using the integral representation of $\sum_{i=1}^L c_i\rho_i$ and  $\sum_{i=1}^L c_i\rho_i^k$ (cf. Formula (\ref{eq:integral})), we get:
$$
\left\{\begin{array}{l}
\sum_{i=1}^L M\left(\hat{c}_i-\frac{N_i}{N}\right)=0,\\
\sum_{i=1}^L M\left(\hat{c_i}\hat{\rho}_i-\frac{N_i}{N}\rho_i\right)=-\frac{M^2}{2N {\bf i} \pi}\oint_{\mathcal{C}}z\left(\frac{\mrnhat'(z)}{\mrnhat(z)}-\frac{\mn'(z)}{\mn(z)}\right)dz,\\
\sum_{i=1}^L M\left(\hat{c}_i \hat{\rho}_i^k-\frac{N_i}{N}\rho_i^k\right)=\frac{M^2(-1)^{k}}{2 {\bf i} (k-1)N\pi}\oint_{\mathcal{C}}\left(\frac{1}{\mrnhat(z)^{k-1}}-\frac{1}{\mn(z)^{k-1}}\right)dz, \hspace{0.1cm} 2\leq k\leq 2L-1.
\end{array}\right.
$$
Denote by $C(\mathcal{C},\mathbb{C})$ the set of continuous functions from $\mathcal{C}$ to $\mathbb{C}$ endowed with the supremum norm $\|u\|_{\infty}=\sup_{\mathcal{C}}|u|$. In the same way as in \cite{YAO11journal}, consider the process:
$
(X_N,X_N',u_N,u_N'):\mathcal{C}\rightarrow \mathbb{C},
$
where
\begin{align*}
X_N(z)&=M\left(\mrnhat(z)-\mn(z)\right),\\
X_N'(z)&=M\left(\mrnhat'(z)-\mn'(z)\right),\\
u_N(z)&=\mrnhat(z), \hspace{0.1cm} u_N'(z)=\mrnhat'(z).
\end{align*}
Then, $M\sum_{i=1}^L \left( \hat{c}_i\hat{\rho}_i-\frac{N_i}{N}\rho_i\right)$ can be written as:
\begin{align*}
M \sum_{i=1}^L \left(\hat{c}_i\hat{\rho}_i-\frac{N_i}{N}\rho_i\right)&=-\frac{M}{2{\bf i}N \pi }\oint_{\mathcal{C}}z\left(\frac{\mn(z)X_N'(z)-u_N'(z)X_N(z)}{\mn(z)u_N(z)}\right)dz,\\
&\triangleq \Upsilon_N(X_N,X_N',u_N,u_N'),
\end{align*}
where
$$
\Upsilon_N(x,x',u,u')=-\frac{M}{2{\bf i} N\pi }\oint_{\mathcal{C}}z\left(\frac{\mn(z)x'(z)-u'(z)x(z)}{\mn(z)u(z)}\right)dz.
$$
On the other hand, using the decomposition $a^k-b^k=(a-b)\sum_{\ell=0}^{k-1}a^{\ell}b^{k-1-\ell}$, we can prove that: 
\begin{align*}
\sum_{i=1}^L M\left(\hat{c}_i\hat{\rho}_i^k-\frac{N_i}{N}\rho_i^k\right)&=\frac{M^2(-1)^{k}}{2{\bf i} N \pi (k-1)}\oint_{\mathcal{C}}\sum_{\ell=0}^{k-2}-\frac{\mrnhat(z)-\underline{m}_N(z)}{\mrnhat^{\ell+1}(z)\underline{m}_N^{k-1-\ell}(z)} dz\\
&=\frac{M(-1)^{k+1}}{2{\bf i} N(k-1) \pi} \oint_{\mathcal{C}}\sum_{\ell=0}^{k-2}X_N(z)u_N(z)^{-\ell-1}\mn(z)^{-k+1+\ell} dz\\
&\triangleq \Phi_{N,k}(X_N,u_N),
\end{align*}
for $2 \leq k \leq {2L-1}$, where
$$
\Phi_{N,k}(x,u)=\frac{M(-1)^{k+1}}{2{\bf i} N(k-1)\pi}\oint_{\mathcal{C}}\sum_{\ell=0}^{k-2}x(z)u(z)^{-\ell-1}\mn(z)^{-k+1+\ell}dz.
$$
The main idea of the proof of the theorem lies in the following steps:
\begin{enumerate}
\item Prove the convergence of $\left[\Upsilon_N(X_N,X_N',u_N,u_N'),\Phi_{N,2}(X_N,u_N),\cdots,\Phi_{N,L}(X_N,u_N) \right]^T$ to a Gaussian random vector with the help of the continuous mapping theorem.
\item Compute the limiting covariance between $M \sum_{i=1}^L \left(\hat{c}_i\hat{\rho}_i^k-\frac{N_i}{N}\rho_i^k \right)$ and $M \sum_{i=1}^L  \left(\hat{c}_i\hat{\rho}_i^{\ell}-\frac{N_i}{N}\rho_i^{\ell} \right)$.
\item Conclude by expressing $M\left[\hat{c}_1-\frac{N_1}{N},\cdots,\hat{c}_L-\frac{N_L}{N},\hat{\rho}_1-\rho_1,\cdots,\hat{\rho}_L-\rho_L\right]^T$ as a linear function of $M\left[\hat{\gamma}_0-\gamma_0,\cdots , \hat{\gamma}_{2L-1}-\gamma_{2L-1}  \right]^T$.
\end{enumerate}

\subsection{Fluctuations of the moments}
The convergence of $\Upsilon_N(X_N,X_N',u_N,u_N')$ to a Gaussian random variable has been established in \cite{YAO11journal}.
It has been proved that:
$$
\Upsilon_N(X_N,X_N',u_N,u_N')\xrightarrow[M,N\to\infty]{\mathcal{D}}\Upsilon(X,Y,\underline{m},\underline{m}')
$$
where
$$
\Upsilon(x,y,v,w)=\frac{1}{2 {\bf i}\pi c }\oint_{\mathcal{C}}z\left(\frac{\underline{m}(z)y(z)-w(z)x(z)}{\underline{m}(z)v(z)}\right)dz.
$$ and $(X,Y)$ is a Gaussian process with mean function zero and covariance function given by:

\begin{align*}
{\rm cov}\left(X(z),X(\tilde{z})\right)&=\frac{\underline{m}'(z)\underline{m}'(\tilde{z})}{\left(\underline{m}(z)-\underline{m}(\tilde{z})\right)^2}-\frac{1}{(z-\tilde{z})^2} \triangleq \kappa(z,\tilde{z}), \\
{\rm cov}\left(Y(z),X(\tilde{z})\right)&= \frac{\partial}{\partial z}\kappa(z,\tilde{z}),\\
{\rm cov}\left(X(z),Y(\tilde{z})\right)&=\frac{\partial}{\partial \tilde{z}}\kappa(z,\tilde{z}),\\
{\rm cov}\left(Y(z),Y(\tilde{z})\right)&=\frac{\partial^2}{\partial z\partial\tilde{z}}\kappa(z,\tilde{z}).
\end{align*}
We also need to prove the convergence in distribution of $\Phi_{N,k}(X_N,u_N)$, for $2\leq k\leq L$. The cornerstone of the proof is the convergence of $X_N:\mathcal{C}\rightarrow \mathbb{C}$ to a Gaussian process $X(z)$ which is ensured in \cite[Lemma 9.11]{SIL06}. Since $u_N\xrightarrow[N,M\to+\infty]{}\underline{m}$, $(X_N,u_N)$ converges in distribution  to $(X,\underline{m})$.

Let $\Phi_k(x,u)$ be defined as:
$$
\Phi_k(x,u)=\frac{(-1)^{k}}{2{\bf i} c \pi}\oint_{\mathcal{C}}x(z)u(z)^{-k}dz.
$$

We want to show that $\Phi_k(X_N,u_N)$ converges in distribution to a Gaussian vector. The continuous mapping theorem is useful to transform one convergence to another.
\begin{proposition}[cf. {\cite[Th. 4.27]{KAL02}}]
\label{th:CMT}
For any metric spaces $S_1$ and $S_2$, let $\xi$, $ (\xi_n)_{n\geq
 1}$ be random elements in $S_1$ with $\xi_n \xrightarrow[n\to
\infty]{\mathcal D} \xi$ and consider some measurable mappings
$f$, $(f_n)_{n\ge 1}$: $S_1 \to S_2$ and a measurable set $\Gamma \subset S_1$
with $\xi \in \Gamma$ a.s. such that $f_n(s_n) \rightarrow f(s)$ as $s_n
\rightarrow s \in \Gamma$. Then $f_n(\xi_n) \xrightarrow[n\to
\infty]{\mathcal D} f(\xi)$.
\end{proposition}

Consider the set:
$$
\Gamma=\left\{(x,u)\in C^2\left(\mathcal{C},\mathbb{C}\right), \inf_{\mathcal{C}}|u|>0\right\}.
$$
Then, since $\inf_{\mathcal{C}}|\underline{m}|>0$ (see \cite[Section 9.12]{SIL06}), the dominated convergence theorem implies that the convergence of $\left(x_N,y_N\right)\rightarrow \left(x, y \right)\in\Gamma$ leads to $\Phi_{N,k}(x_N,y_N)\rightarrow \Phi_k(x, y)$. The continuous mapping theorem applies, thus giving:
$$
\Phi_{N,k}(X_N,u_N)\xrightarrow[M,N\to\infty]{\mathcal{D}}\Phi_k(X,u).
$$
It now remains to prove that the limit law $\Phi_k(X,u)$ is Gaussian. For that, it suffices to notice that the integral can be written as the limit of a finite Riemann sum and that a finite Riemann sum of the elements of a Gaussian random vector is still Gaussian.

The convergence of $\Upsilon_N(X_N,X_N',u_N,u_N')$ and $\Phi_{N,k}(X_N,u_N)$ is not sufficient to conclude about that of the whole vector. The additional requirement is to prove the convergence to a Gaussian distribution of any linear combination of  $\left[\Upsilon_N(X_N,X_N',u_N,u_N'),\Phi_{N,2}(X_N,u_N),\cdots,\Phi_{N,L}(X_N,u_N) \right]^T$, which can be easily established in the same way as before.
It implies that this vector converges to a Gaussian vector. This ends the first step of the proof. 

\subsection{Computation of the variance}
We now come to the second step. We shall therefore evaluate the quantities:
\begin{align*}
{\bf V}_{1,1}&=\mathbb{E}\left[\Upsilon(X,Y,\underline{m},\underline{m}')\Upsilon(X,Y,\underline{m},\underline{m}')\right],\\
{\bf V}_{1,k}&={\bf V}_{k,1}=\mathbb{E}\left[\Upsilon(X,Y,\underline{m},\underline{m}')\Phi_k(X,\underline{m})\right],\hspace{0.5cm} 2\leq k\leq L,\\
{\bf V}_{k,\ell}&=\mathbb{E}\left[\Phi_k(X,\underline{m})\Phi_{\ell}(X,\underline{m})\right], \hspace{0.5cm}  2\leq k, \ell\leq 2L-1.
\end{align*}
The details of the calculations are in Appendix \ref{app:variance} and yield:
\begin{equation}
{\bf V}_{k,\ell}=-\frac{(-1)^{k+\ell}}{4\pi^2c^2}\oint_{\mathcal{C}_1}\oint_{\mathcal{C}_2}
\left[\frac{\underline{m}'(z_1)\underline{m}'(z_2)}{\left(\underline{m}(z_1)-\underline{m}(z_2)\right)^2}-\frac{1}{(z_1-z_2)^2}\right]\frac{1}{\underline{m}^k(z_1)\underline{m}^{\ell}(z_2)}dz_1 dz_2\ ,
\label{eq:Vkl}
\end{equation}
for $1\leq k,\ell \leq 2L-1$. Let ${\bf
  w}_M=M\left[\hat{\gamma}_0-\gamma_0, \cdots,
  \hat{\gamma}_{2L-1}-\gamma_{2L-1} \right]^{T}$.

We have just proved that vector ${\bf w}_M$ converges asymptotically to:
$$
{\bf w}_M\xrightarrow[N,M\to+\infty]{\mathcal{D}}\mathcal{N}_{2L}(0,{\bf W}),
$$
where
$$
{\bf W}=\left[
\begin{BMAT}{cc}{cc}
0& {\bf 0}\\
{\bf 0} & {\bf V}
\end{BMAT}\right]
$$
and ${\bf V}$ is the $(2L-1)\times (2L-1)$ matrix whose entries $V_{k,l}$ are given by \eqref{eq:Vkl}.
\begin{remark}
The zeros in the variance are simply from the fact that $\hat{\gamma}_0-\gamma_0=0$.

\end{remark}
\subsection{Fluctuations of the  eigenvalues estimates}

To transfer this convergence to ${\bf q}_M\triangleq M\left[\hat{c}_1-\frac{N_1}{N},\cdots,\hat{c}_L-\frac{N_L}{N},\hat{\rho}_1-\rho_1,\cdots,\hat{\rho}_L-{\rho}_L\right]^{T}$, we shall use Slutsky's lemma which is as below:

\begin{lemma}[cf. \cite{VAN00}] Let ${\bf X}_n$, ${\bf Y}_n$ be sequences of vector or matrix random elements. If ${\bf X}_n$ converges in distribution to a random element $ {\bf X}$, and ${\bf Y}_n$ converges in probability to a constant ${\bf C}$, then

$${\bf Y}_n^{-1} {\bf X}_n \xrightarrow[]{\mathcal D} {\bf C}^{-1} {\bf X}$$
provided that ${\bf C}$ is invertible.

\end{lemma}

We will show that ${\bf w}_M$ satisfies the following linear system:

\begin{equation}\label{eq:slutsky}
{\bf w}_M = \hat{\bf M}_M  {\bf q}_M
\end{equation}
where $\hat{\bf M}_M$ converges in probability to ${\bf M}$ which is given by

$$
{\bf M}=\begin{bmatrix}
1 &\cdots & 1 &0 &\cdots &0 \\
\rho_1 & \cdots &\rho_L & c_1 & \cdots & c_L\\
\vdots & \ddots & \ddots & \ddots & \ddots & \vdots \\
\rho_1^{2L-1} &\cdots & \rho_L^{2L-1} & (2L-1)c_1\rho_1^{2L-2}& \cdots & (2L-1)c_L\rho_L^{2L-2}
\end{bmatrix}.
$$

To this end, let us work out the expression of $w_{k,M}$, the $k$-th element of ${\bf w}_M$.

If $k=1$, it is easy to see that $w_{1,M}=0$.

For $k\geq 2$, $w_{k,M}$ is given by:
\begin{align*}
w_{k,M}&=M\sum_{i=1}^L\left(\hat{c}_i\hat{\rho}_i^{k-1}-\frac{N_i}{N}{\rho}_i^{k-1}\right)\\
&=M\sum_{i=1}^L\left(\hat{c}_i\hat{\rho}_i^{k-1}-\frac{N_i}{N}\hat{\rho}_i^{k-1}+\frac{N_i}{N}\hat{\rho}_i^{k-1}-\frac{N_i}{N}\rho_i^{k-1}\right)\\
&=M\sum_{i=1}^L \left(\left(\hat{c}_i-\frac{N_i}{N}\right)\hat{\rho}_i^{k-1}+\frac{N_i}{N}(\hat{\rho}_i-\rho_i)\sum_{\ell=0}^{k-2}\hat{\rho}_i^{\ell}\rho_i^{k-2-\ell}\right).\\
\end{align*}

Then define
$$\hat{\bf M}_M=
\left(\begin{matrix}
1 & \cdots & 1&0& \cdots & 0\\
\hat{\rho}_1 & \cdots & \hat{\rho}_L & \frac{N_1}{N} & \cdots & \frac{N_L}{N} \\
\vdots&\ddots&\ddots & \ddots&\ddots& \vdots\\
\hat{\rho}_1^{2L-1}&\cdots & \hat{\rho}_L^{2L-1}&  \frac{N_1}{N} \sum_{\ell=0}^{2L-2}\hat{\rho}_1^\ell\rho_1^{2L-2-\ell} & \cdots & \frac{N_L}{N} \sum_{\ell=0}^{2L-2}\hat{\rho}_L^{\ell}\rho_L^{2L-2-\ell} \\
\end{matrix} \right).$$

We can see easily that the equation (\ref{eq:slutsky}) is satisfied and $\hat{\bf M}_M$ converges in probability to $\bf M$. It remains to check that ${\bf M}$ is invertible. Note that the non-singularity of matrix ${\bf M}$ has been already established in Section \ref{sec:estimation}, where this property was required to prove the existence of an estimator.
As a consequence, using Slutsky's lemma, we deduce that:
$$
\hat{\bf M}_M{\bf q}_M\xrightarrow[M,N\to+\infty]{\mathcal{D}}\mathcal{N}_{2L}(0,{\bf W})
$$
and

$${\bf q}_M \xrightarrow[M,N\to+\infty]{\mathcal{D}} \mathcal{N}_{2L}\left( 0, {\bf M}^{-1} {\bf W} ({\bf M}^{-1})^T \right).$$
This ends the proof for the fluctuation.
\end{proof}

\section{Simulations}
\label{sec:simulations}

In this section, we compare the performance of the proposed method with that of Mestre's estimator in \cite{MES08b}. We also verify by simulations the accuracy of the Gaussian approximation stated by the Central Limit theorem.

In the first experiment, we consider a covariance matrix ${\bf R}_N$ with three different eigenvalues $\left(\rho_1,\rho_2,\rho_3\right)=(1,3,5)$ uniformly distributed i.e, $\frac{N_1}{N}=\frac{N_2}{N}=\frac{N_3}{N}=\frac{1}{3}$. 
We set the ratio between the number of samples and the number of variables $\frac{N}{M}$ to $3/8$, a situation for which the separability does not obviously hold (see Fig. \ref{fig:clusters2}). Since the knowledge of the multiplicities is available when using the estimator in \cite{MES08b}, we assume, the same for the proposed method. Hence, the estimation of the polynomial whose roots are $\rho_i$ could not be as described previously. It is actually performed using the  Newton-Girard formulas, which relates the coefficients of a  polynomial to the power sum of its roots.


We compare the performance of both estimators for different values of  $M$ and $N$ satisfying a constant ratio $c=N/M=3/8$. Fig \ref{fig:comparison}, the experienced mean square error (MSE)  in the estimation process for each method, where the MSE is given by:
$$
{\rm MSE}\triangleq \sum_{i=1}^3 |\hat{\rho}_i-\rho_i|^2.
$$
\begin{minipage}[c]{0.43\textwidth}
\begin{figure}[H]
  \begin{center}
  \begin{tikzpicture}[font=\footnotesize,scale=0.75]
    \renewcommand{\axisdefaulttryminticks}{8}
    \tikzstyle{every major grid}+=[style=densely dashed]
    \tikzstyle{every pin}=[fill=white,draw=black,font=\footnotesize,edge style={<-}]
    \tikzstyle{every axis x label}+=[yshift=5pt]
   \tikzstyle{every axis legend}+=[cells={anchor=west},fill=white,
       at={(0.98,0.98)}, anchor=north east, font=\scriptsize ]

    \begin{axis}[
      grid=major,
      xmajorgrids=true,
      ymajorgrids=true,
      xlabel={$N$},
      ylabel={MSE in dB},
      xmin=30,
      xmax=150,
      ymin=-22,
      ymax=-6,
      ]
	  \addplot[smooth,blue,line width=0.5pt,mark=x]plot coordinates{(30.000000,-9.507336)(60.000000,-11.647864)(90.000000,-12.176486)(120.000000,-12.285698)(150.000000,-12.436594)
};
      \addplot[smooth,red,line width=0.5pt,mark=o] plot coordinates{
	  	 (30.000000,-6.861267)(60.000000,-13.765248)(90.000000,-17.594329)(120.000000,-19.574075)(150.000000,-21.617828)	  };
	 \legend{{Mestre Estimator},{Proposed Estimator}}
\end{axis}
\end{tikzpicture}
\caption{Experienced MSE  with $N$ when $\frac{N}{M}=\frac{3}{8}$ and $(\rho_1,\rho_2,\rho_3)=(1,3,5)$}
\label{fig:comparison}
\end{center}
\end{figure}
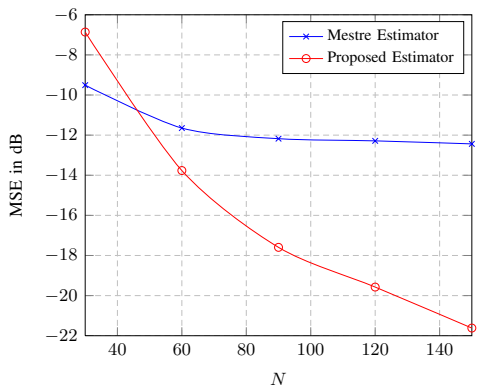
\end{minipage}\hfill
\begin{minipage}[c]{0.43\textwidth}
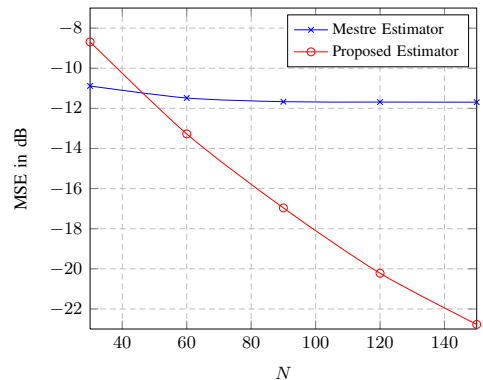
\begin{figure}[H]
  \begin{center}
  \begin{tikzpicture}[font=\footnotesize,scale=0.75]
    \renewcommand{\axisdefaulttryminticks}{8}
    \tikzstyle{every major grid}+=[style=densely dashed]
    \tikzstyle{every pin}=[fill=white,draw=black,font=\footnotesize,edge style={<-}]
    \tikzstyle{every axis x label}+=[yshift=5pt]
    \tikzstyle{every axis legend}+=[cells={anchor=west},fill=white,
        at={(0.98,0.98)}, anchor=north east, font=\scriptsize ]

    \begin{axis}[
     grid=major,
      xmajorgrids=true,
      ymajorgrids=true,
      xlabel={$N$},
      ylabel={MSE in dB},
      xmin=30,
      xmax=150,
      ymin=-23,
      ymax=-7,
      ]
	   \addplot[smooth,blue,line width=0.5pt,mark=x]plot coordinates{(30.000000,-10.883966)(60.000000,-11.483270)(90.000000,-11.667406)(120.000000,-11.685986)(150.000000,-11.692651)
};
      \addplot[smooth,red,line width=0.5pt,mark=o] plot coordinates{
	  	  (30.000000,-8.688546)(60.000000,-13.275337)(90.000000,-16.964299)(120.000000,-20.221400)(150.000000,-22.770789)
	  };
	 \legend{{Mestre Estimator},{Proposed Estimator}}
\end{axis}
\end{tikzpicture}
\caption{Experienced MSE  with $N$ when $\frac{N}{M}=\frac{3}{8}$ and $(\rho_1,\rho_2,\rho_3)=(1,1.5,2)$}
\end{center}
\label{fig:comparison2}
\end{figure}
\end{minipage}

We note that as $M$ and $N$ increase, the estimator in \cite{MES08b} exhibits an error floor since the separability condition is not satisfied and thus is no longer consistent. We also conduct the same experiment when $\rho_1$, $\rho_2$ and $\rho_3$ are set respectively to $1,1.5$ and $2$. We note that in this case, the asymptotic gap with Mestre's estimator is further large (See Fig. 4).

In the second experiment, we verify by simulations the accuracy of the Gaussian approximation. We consider the case where there are two different eigenvalues $\rho_1=1$ and $\rho_2=3$ that are uniformly distributed. Unlike the first experiment, we assume that the multiplicities are not knwon. 
We represent in Fig \ref{fig:clt} the histogram for $\hat{\rho}_1$ and $\hat{\rho}_2$ when $N=60$ and $M=120$.
We also represent in red line, the corresponding Gaussian distribution. We note that as it was predicted by our derived results,
the histogram is similar to that of a Gaussian random variable.

\begin{figure}
\begin{center}
\includegraphics[width=0.5\textwidth]{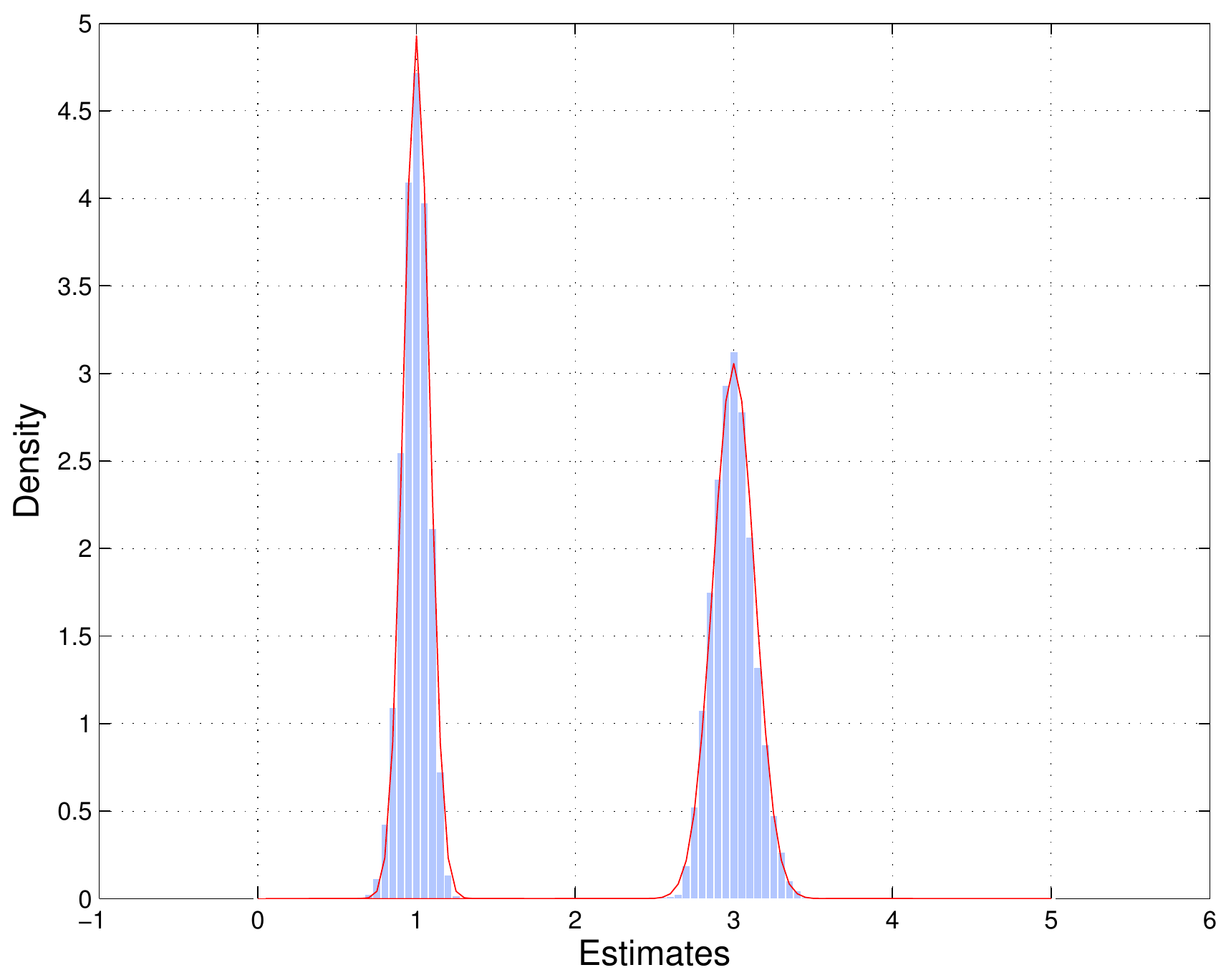}
\end{center}
\caption{Comparison of empirical against theoretical variances for $c_1=c_2=0.5$ and $\rho_1=1$ and $\rho_2=3$}
\label{fig:clt}
\end{figure}

\section{Discussion}\label{sec:discuss}

The present work is a theoretical contribution to the important
problem of estimating the covariance matrices of large dimensional
data. Two important assumptions (separability condition, exact
knowledge of the multiplicity) have been in particular relaxed with
respect to previous work. From a numerical point of view, it should be
noticed however, that the situation is more contrasted: If the
multiplicities are known, previous simulations show good performance;
if not, then one needs to enlarge the dimension of the observations to
achieve a good performance. Moreover, if the eigenvalues of ${\bf
  R}_N$ are far away from each other, then only the largest eigenvalue is
well-estimated because in the expression of the moments, the term
corresponding to the largest eigenvalue prevails. On the other hand,
if the eigenvalues are too close to each other, matrix
$\boldsymbol{\Gamma}$ is ill-conditioned, thus enlarging the induced
error. These phenomenas are inherent to the moment method, and
preliminary studies show that using trigonometric moments might help
mitigating these numerical problems.


\label{sec:conclusion}
\appendices
\section{Proof of lemma \ref{lem:moment}}
\label{app:proof_lemma}
By Cauchy's formula, write:
$$
\sum_{k=1}^L \frac{N_k}{N} \rho_k^{\ell}=\frac{1}{2{\bf i}\pi N}\oint_{\Gamma}\sum_{r=1}^L \frac{N_r\omega^{\ell}}{\omega-\rho_r} d\omega,
$$
where $\Gamma$ is a counterclockwise oriented contour that circles all eigenvalues $\{\rho_1,\cdots, \rho_L\}$. Performing the changing variable $\omega=-\frac{1}{\underline{m}_N(z)}$ in the same manner as in \cite{MES08b}, we get:
$$
\sum_{k=1}^L \frac{N_k}{N} \rho_k^{\ell}=\frac{(-1)^{\ell+1}}{2{\bf i}\pi N}\oint_{\mathcal{C}} \sum_{r=1}^L \frac{N_r\underline{m}_N'(z) dz}{\underline{m}_N^{\ell+1}(z)\left(\rho_r\underline{m}_N(z)+1\right)},
$$
where the contour ${\mathcal C}$ is counterclockwise oriented which contains the whole support ${\mathcal S}$.

From \eqref{eq:equivalent}, we can establish that:
$$
m_N(z)=-\frac{1}{Nz}\sum_{r=1}^L \frac{N_r}{1+\rho_r\underline{m}_N(z)},
$$
thus yielding:
\begin{equation}
\sum_{k=1}^L \frac{N_k}{N}\rho_k^{\ell}=\frac{(-1)^{\ell}}{2{\bf i}\pi}\oint_{\mathcal{C}}\frac{z\underline{m}_N'(z)}{\underline{m}_N^{\ell+1}(z)}m_N(z) dz.
\label{eq:previous}
\end{equation}
Plugging the relation:
$$
m_N(z)=\frac{M}{N}\underline{m}_N(z)+\frac{M(1-\frac{N}{M})}{Nz}
$$
into \eqref{eq:previous}, we obtain:
\begin{equation}
\sum_{k=1}^L \frac{N_k}{N}\rho_k^{\ell}=\frac{(-1)^{\ell}}{2{\bf i}\pi}\oint_{\mathcal{C}}\frac{Mz\underline{m}_N'(z) dz}{N\underline{m}_N^{\ell}(z)}+\frac{(-1)^{\ell}}{2{\bf i}\pi}\oint_{\mathcal{C}}\frac{M(1-\frac{N}{M})\underline{m}_N'(z)}{N\underline{m}_N^{\ell+1}(z)}dz.
\label{eq:fund}
\end{equation}
Since $\frac{\underline{m}_N'(z)}{\underline{m}_N^{\ell+1}(z)}$ is the derivative of $-\frac{1}{\ell\underline{m}_N^{\ell}(z)}$,
$$
\oint_{\mathcal{C}}\frac{\underline{m}_N'(z)}{\underline{m}_N^{\ell+1}(z)}dz=0.
$$
The second term on the right hand side of \eqref{eq:fund} is then equal to zero. It remains thus to deal with $\oint_{\mathcal{C}} \frac{z\underline{m}_N'(z)}{\underline{m}_N^{\ell}(z)}$. If $\ell\geq 2$, by integration by parts, we obtain:
$$
\oint_{\mathcal{C}}\frac{z\underline{m}_N'(z)}{\underline{m}_N^{\ell}(z)}dz =\frac{1}{\ell-1}\oint_{\mathcal{C}}\frac{dz}{\underline{m}_N^{\ell-1}(z)}.
$$
We thus obtain:
\begin{equation} \label{eq:integral}
\sum_{k=1}^L \frac{N_k}{N}\rho_k^{\ell}=\frac{M(-1)^{\ell}}{2{\bf i}\pi N (\ell-1)}\oint_{\mathcal{C}}\frac{dz}{\underline{m}_N^{\ell-1}(z)}.
\end{equation}
Finally, we propose to substitute the unknown term $\underline{m}_N(z)$ by its asymptotic equivalent $\mrnhat(z)$. Let $\hat{\gamma}_0,\cdots,\hat{\gamma}_{2L-1}$ the real quantities given by:
$$
\begin{array}{lll}
 \hat{\gamma}_0&=&1,\\
\hat{\gamma}_1&=&-\frac{M}{2N{\bf i}\pi}\oint_{\mathcal{C}}\frac{z\mrnhat'(z)}{\mrnhat(z)}dz,\\
&\vdots &\\
\hat{\gamma}_{2L-1}&=&\frac{M(-1)^{2L-1}}{2N(2L-1){\bf i}\pi}\oint_{\mathcal{C}}\frac{dz}{\mrnhat^{2L-1}(z)}.
 \end{array}
 $$
Then, by the dominated convergence theorem and the fact that with probability one \cite[Section 9.12]{SIL06}, $$
\inf_{z\in\mathcal{C}}|\underline{m}_N(z)|>0
$$
and
 $$
\inf_{z\in\mathcal{C}}|\mrnhat(z)|>0,
$$
one obtains: for all $k\geq 2$,
$$\left| \int_{\mathcal C} \frac{dz}{\underline{m}_N^{k-1}(z)}- \int_{\mathcal C} \frac{dz}{m_{\underline{\hat{\bf R}}_N}^{k-1}(z)} \right| \xrightarrow[]{a.s.} 0$$ and
$$\left| \int_{\mathcal C} \frac{m_{\underline{\hat{\bf R}}_N}'(z)dz}{m_{\underline{\hat{\bf R}}_N}(z)} - \int_{\mathcal C} \frac{\underline{m}_N'(z)dz}{\underline{m}_N(z)} \right| \xrightarrow[]{a.s.} 0.$$
Consequently:
$$
\hat{\gamma}_i-\gamma_i\xrightarrow[N,M\to\infty]{a.s.} 0.
$$

\section{Calculation of the variance}
\label{app:variance}
In this section, we will show the calculations of the variance matrix ${\bf V}$. The computation of ${\bf V}_{1,1}$ has been carried out in \cite{YAO11journal} where it was shown that:
$$
{\bf V}_{1,1}= -\frac{1}{4\pi^2c^2}\oint_{\mathcal{C}_1}\oint_{\mathcal{C}_2} \left[\frac{\underline{m}'(z_1)\underline{m}'(z_2)} {\left(\underline{m}(z_1)-\underline{m}(z_2)\right)^2}-\frac{1}{(z_1-z_2)^2}\right] \frac{1}{\underline{m}(z_1)\underline{m}(z_2)}dz_1 dz_2,
$$
with ${\mathcal C_1}$ and ${\mathcal C_2}$ defined in the theorem.
Using the fact that $\inf_{z\in\mathcal{C}}|\underline{m}(z)|>0$ together with Fubini's theorem, the quantity ${\bf V}_{k,\ell}$ for $k\geq 2, \ell\geq 2$, becomes:
$$
{\bf V}_{k,\ell}=-\frac{(-1)^{k+\ell}}{4\pi^2c^2}\oint_{\mathcal{C}_1}\oint_{\mathcal{C}_2}\mathbb{E}\left[X(z_1)X(z_2)\right]\underline{m}^{-k}(z_1)\underline{m}^{-\ell}(z_2)dz_1 dz_2.
$$
Substituting $\mathbb{E}\left[X(z_1)X(z_2)\right]$ by $\kappa(z_1,z_2)$, we obtain:
$$
{\bf V}_{k,\ell}=-\frac{(-1)^{k+\ell}}{4\pi^2c^2}\oint_{\mathcal{C}_1}\oint_{\mathcal{C}_2}
\left[\frac{\underline{m}'(z_1)\underline{m}'(z_2)}{\left(\underline{m}(z_1)-\underline{m}(z_2)\right)^2}-\frac{1}{(z_1-z_2)^2}\right]\frac{1}{\underline{m}^k(z_1) \underline{m}^{\ell}(z_2)}dz_1 dz_2.
$$
Finally, it remains to compute ${\bf V}_{k,1}$. Expanding $\Upsilon(X,Y,\underline{m},\underline{m}')$ and $\Phi_k(X,\underline{m})$, we obtain:
\begin{align*}
{\bf V}_{k,1}&= -\frac{(-1)^{k+1}}{4\pi^2c^2} \oint_{\mathcal{C}_1}\oint_{\mathcal{C}_2} \left[ \frac{z_2}{\underline{m}(z_2)\underline{m}^{k}(z_1)}\mathbb{E}\left[X(z_1)X'(z_2)\right]dz_1dz_2- \frac{\underline{m}'(z_2)}{\underline{m}(z_2)^2\underline{m}^{k}(z_1)}\mathbb{E}\left[X(z_1)X(z_2)\right] \right] dz_1 dz_2  \\
&=-\frac{(-1)^{k+1}}{4\pi^2c^2}\left(\oint_{\mathcal{C}_1}\oint_{\mathcal{C}_2}\frac{z_2\partial_2\kappa(z_1,z_2)}{\underline{m}(z_2)\underline{m}(z_1)^{k}}dz_1dz_2-\oint_{\mathcal{C}_1}\oint_{\mathcal{C}_2} \frac{\underline{m}'(z_2)\kappa(z_1,z_2)}{\underline{m}^2(z_2)\underline{m}^k(z_1)}dz_1 dz_2\right).
\end{align*}
By integration by parts, we obtain:
$$
\oint_{\mathcal{C}_2}\frac{z_2\partial_2\kappa(z_1,z_2)}{\underline{m}(z_2)\underline{m}^k(z_1)} dz_2= -\oint_{\mathcal{C}_2}\frac{\kappa(z_1,z_2)}{\underline{m}(z_2)\underline{m}^k(z_1)}dz_2+\oint_{\mathcal{C}_2}\frac{\underline{m}'(z_2)\kappa(z_1,z_2)}{\underline{m}(z_2)^2\underline{m}^k(z_1)}dz_2.
$$
Hence,
$$
{\bf V}_{k,1}=-\frac{(-1)^{k+1}}{4\pi^2c^2}\oint_{\mathcal{C}_1}\oint_{\mathcal{C}_2}\frac{\kappa(z_1,z_2) dz_1 dz_2}{\underline{m}(z_2)\underline{m}^k(z_1)}.
$$
This extends the expression of ${\bf V}_{k,l}$ for any $k,\ell \in\left\{1,\cdots,L-1\right\}$, thus yielding:
\begin{equation}
{\bf V}_{k,\ell}=-\frac{(-1)^{k+\ell}}{4\pi^2c^2}\oint_{\mathcal{C}_1}\oint_{\mathcal{C}_2}
\left[\frac{\underline{m}'(z_1)\underline{m}'(z_2)}{\left(\underline{m}(z_1)-\underline{m}(z_2)\right)^2}-\frac{1}{(z_1-z_2)^2}\right]\frac{1}{\underline{m}^k(z_1)\underline{m}^{\ell}(z_2)}dz_1 dz_2.
\end{equation}

\vspace{8mm}

\bibliography{IEEEabrv,IEEEconf,tutorial_RMT}

\end{document}